\documentclass{article}

\usepackage{arxiv}

\usepackage[utf8]{inputenc} 
\usepackage[T1]{fontenc}    
\usepackage{fancyhdr}
\usepackage{lmodern}
\usepackage[colorlinks]{hyperref}  
\usepackage{url}            
\usepackage{booktabs}       
\usepackage{amsfonts}       
\usepackage{nicefrac}       
\usepackage{microtype}      
\usepackage{lipsum}

\usepackage{graphicx}
\usepackage{bm}
\usepackage{epsfig}
\usepackage{epstopdf}
\usepackage{subfigure}

\usepackage{enumerate}
\usepackage{amssymb}
\usepackage{amsmath}
\usepackage{amsthm}
\usepackage{enumitem}
\usepackage{mathrsfs}

\newcommand{\R}{\mathbb{R}}
\newcommand{\C}{\mathbb{C}}
\newcommand{\N}{\mathbb{N}}

\renewcommand{\bar}{\overline}
\renewcommand{\rho}{\varrho}
\renewcommand{\phi}{\varphi}

\newcommand{\Ra}{\mathrm{Ra}}
\newcommand{\Da}{\mathrm{Da}}

\DeclareMathOperator{\esssup}{ess\,sup}

\newcommand{\scal}[2]{\left\langle#1,#2\right\rangle}

\theoremstyle{plain}
\newtheorem{thm}{Theorem}[section]

\newtheorem{lem}{Lemma}[section]

\title{A Darcy-Brinkman Model for Penetrative Convection in LTNE}

\author{  
G. Arnone\href{https://orcid.org/0000-0002-3317-6358}{\includegraphics[scale=0.1]{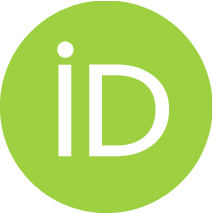}} \\ Dipartimento di Matematica e Applicazioni 'R.Caccioppoli' \\ Università degli Studi di Napoli Federico II \\ Via Cintia, Monte S.Angelo, 80126 Napoli \\ Italy \\   
\texttt{giuseppe.arnone@unina.it} \\ 
\And
F. Capone\thanks{Corresponding author.} \href{https://orcid.org/0000-0002-0672-999X}{\includegraphics[scale=0.1]{orcid.eps}} \\ Dipartimento di Matematica e Applicazioni 'R.Caccioppoli' \\ Universit\'a degli Studi di Napoli Federico II \\ Via Cintia, Monte S.Angelo, 80126 Napoli \\ Italy \\ 
\texttt{fcapone@unina.it} \\
\And 
J.A. Gianfrani\href{https://orcid.org/0000-0001-9906-2495}{\includegraphics[scale=0.1]{orcid.eps}} \\ Research Centre for Fluid and Complex Systems \\ Coventry University \\ Priory St, Coventry CV1 5FB \\ United Kingdom \\   
\texttt{ae2688@coventry.ac.uk} }

\renewcommand{\shorttitle}{A Darcy-Brinkman Model for Penetrative Convection in LTNE}

\begin{document}
\maketitle

\begin{abstract}
The aim of this paper is to investigate the onset of penetrative convection in a Darcy-Brinkmann
porous medium under the hypothesis of local therma non-equilibrium. For the problem at stake, the strong form of the principle of exchange of stabilities has been proved, i.e. convective motions can occur only through a secondary stationary motion. We perform linear and nonlinear stability analyses of the basic state motion, with particular regard to the behaviour of the stability thresholds with respect to the relevant physical parameters characterizing the model. 
    The Chebyshev-$\tau$ method and the shooting method are employed and implemented to solve the differential eigenvalue problems arising from linear and nonlinear analyses to determine critical Rayleigh numbers. Numerical simulations prove the stabilising effect of upper bounding plane temperature, Darcy's number and the interaction coefficient characterising the local thermal non-equilibrium regime. 
\end{abstract}

\section{Introduction}

In the present paper, we analyse the onset of penetrative convective currents in a high-porosity porous medium in local thermal non-equilibrium (LTNE). Unlike natural convection, this type of convection occurs when the density of the fluid is not monotonically decreasing in temperature (as for the majority of fluids) but the fluid contracts, rather than expands, through heating.
This anomalous behaviour, sometimes called \emph{density inversion phenomenon}, is exhibited in very few cases in nature, among which water represents the most important case. Other cases include graphene, some complex compounds, some iron alloys, and cubic zirconium tingstenate ($\text{Zr}\text{W}_2\text{O}_8$), \cite{bettini2016course}. In the following, we will consider a porous layer $\Sigma$ saturated by water with the bottom plane maintained at $0^\circ C$ and the top plane at a temperature greater than $4^\circ C$ (which is the temperature's value at which water has a maximum). It turns out that $\Sigma$ is partitioned in two sub-regions: $\Sigma_1$ and $\Sigma_2$, as depicted in Figure \ref{pensetup}. The former region is delimited by the lower plane and an intermediate plane where the fluid density attains its maximum, and the latter is delimited by the intermediate plane and the upper plane. Because of the above-described anomalous behaviour, the assumed physical setup implies that $\Sigma_1$ is a potentially unstable fluid region laying below a stably stratified region $\Sigma_2$. When convective motion occur in $\Sigma_1$, the fluid motion will then penetrate in the upper stable region $\Sigma_2$. The first constitutive equation for the density-temperature relation which describes the penetrative convection phenomenon was proposed by Veronis in 1963, \cite{Veronis1963}. Since the pioneering piece of work by Veronis appeared, considerable attention has been devoted to penetrative convection, see e.g. \cite{Straugh1985,CapGentHill2010,arnone2023onset}, in particular because of its applications in geophysics \cite{GeorGunnStraugh1989,HuttStraugh1997}. These papers model the penetrative convection phenomenon in porous media under the hypothesis of local thermal equilibrium (LTE), by defining a single temperature for both fluid and solid phases, which are in thermal equilibrium. From the modeling viewpoint, one temperature equation is required. The novelty of the present paper is the development of a theoretical investigation of a two-field penetrative convection model which involves two energy equations, one for the fluid phase and one for the solid one. These equations are coupled by means of a term representing loss or gain of heat from the other phase.
In literature, assuming that heat exchanges between fluid and solid phases are allowed is well-known as LTNE hypothesis. The two-temperature model was first introduced by Banu and Rees with a groundbreaking paper \cite{banu2002onset} in 2002. Since then, many researchers have devoted their attention to thermal instability phenomena in porous media in LTNE \cite{barletta2015instability, barletta2011local, kuznetsov2011effect}. Straughan in 2006 \cite{straughan2006global} was the first to study the nonlinear stability in a porous medium in LTNE. Interesting results are obtained when considering Cattaneo's law as constitutive equation for heat diffusion in the solid phase \cite{capone2022onset, straughan2013porous, hema2021cattaneo}. 
\begin{figure}[h!]
    \centering
    \includegraphics[scale=0.4]{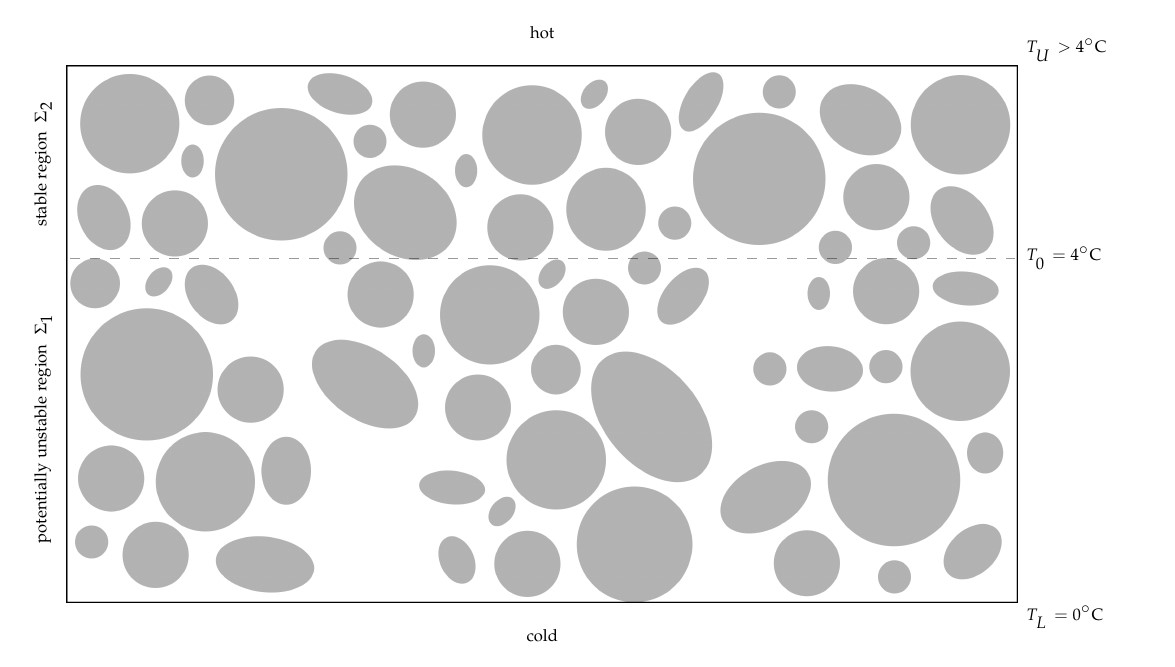}
    \caption{Physical setup for the penetrative convection problem.}
    \label{pensetup}
\end{figure}
The paper is organised as follows. In Section 2 the mathematical model describing the phenomenon of penetrative convection in the hypothesis of LTNE is described and the non-dimensional equations governing the evolution of the perturbations of the conduction solution are derived. In Section 3 the strong form of the principle of exchange of stabilities is proved, and the instability analysis of the linearised model is performed in order to determine the instability thresholds for the onset of penetrative convection, via the Chebyshev-$\tau$ method. In Section 4, weighted energy analysis is performed and the nonlinear critical Rayleigh numbers are determined via the shooting method. In Section 5 numerical simulations are performed, commented on, and summarised. The paper ends with an Appendix containing several details about the employed Chebyshev-$\tau$ method and some technical details about a useful employed apriori estimate in the context of nonlinear analysis.

\section{Mathematical Model}\label{model}
Let $Oxyz$ be a Cartesian frame of reference where the $z$-axis is vertically upward and let us consider a horizontal isotropic porous layer delimited by two impervious planes ($z=0$ and $z=d$) and saturated by a Newtonian fluid at rest. Planes confining the layer are kept at a constant temperature at any time in such a way that the fluid-saturated porous medium is heated from above. As a consequence, a uniform gradient of temperature is imposed and maintained constant across the medium.
Moreover, let us assume the LTNE between solid matrix and fluid, which implies that heat exchanges between solid and fluid phases occur. This assumption requires us to define two unknown fields: the fluid temperature $T^f$ and the solid temperature $T^s$. Hence, let $T_L$ be the uniform temperature on the lower plane ($z=0$) and let $T_U$ be the uniform temperature on the upper plane ($z=d$). The following boundary conditions hold for the problem at stake:
\begin{equation}\label{bc_iniz}
    T^s=T^f=T_L \ \text{on } z=0, \qquad T^s=T^f=T_U \ \text{on } z=d,
\end{equation}
where $T_U>T_L$. In penetrative convection problems, the planes' temperatures are prescribed in such a way that fluid density $\rho_f(T^f)$ may achieve a maximum value in $[T_L,T_U]$. Specifically, as shown by Veronis \cite{Veronis1963}, fluid density exhibits a parabolic behaviour in that temperature interval when it is described by the quadratic function  
\begin{equation}\label{quadraticdensity}
    \rho_f(T^f)=\rho_0\left[1-\alpha(T^f-T_0)^2\right],
\end{equation}
where $T_0$ is a reference temperature, $\rho_0$ the corresponding reference fluid density and $\alpha$ is the thermal expansion coefficient. We now restrict our attention to when the fluid saturating porous medium is \emph{water}. Indeed, among all fluids, water exhibits this density inversion phenomenon in the neighborhood of $4^\circ C$. As shown in Figure \ref{figVer}, the constitutive density law \eqref{quadraticdensity} by Veronis, when $T_0=4^\circ C$, accurately catches the experimental data reported in Table \ref{tabVer}.
\begin{figure}[h!]
    \centering
    \includegraphics[scale=0.7]{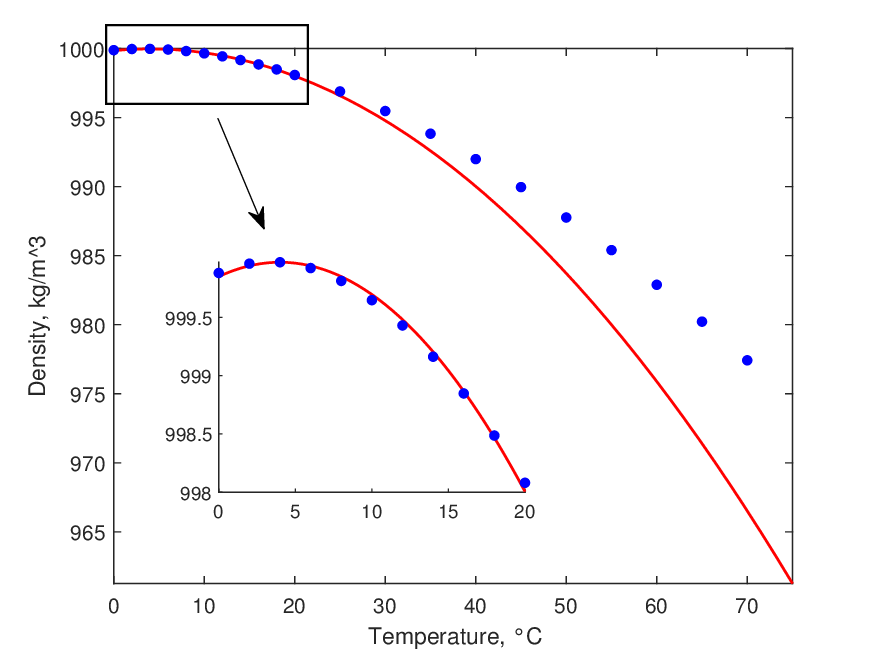}
    \caption{Comparison between experimental values of the water density under isobaric condition \cite{wagner2002iapws} and the Veronis' constitutive law \eqref{quadraticdensity} where $\alpha=7.68\times 10^{-6}\phantom{.}^\circ C^{-2}$ and $T_0=4^\circ C$.}
    \label{figVer}
\end{figure}
\begin{table}[h!]\centering
 \scriptsize
		\begin{tabular}{@{}lll@{}}\toprule
	water density $(kg/m^3)$  &&$ T_U(^\circ C) $\\ \midrule
 999.880 && 0 \\
 999.960 && 2 \\
  999.972  && 4 \\
  999.922    && 6 \\
   999.812     && 8 \\
   999.647       && 10 \\
			\bottomrule
		\end{tabular}
		\caption{Experimental data of water density.}
		\label{tabVer}
	\end{table}\\

Let us introduce a very last assumption regarding the porous layer. In order for the LTNE hypothesis to hold strongly we assume that the porous medium exhibits very high porosity \cite{straughan2015convection}. As a consequence, within a representative elementary volume, the fraction of void spaces is considerably greater than the total volume. 

It is well-known in literature that in order to better describe the fluid motion within a porous medium with high porosity the Darcy-Brinkman law is the best candidate, \cite{NB}. According to this law, the momentum equation is 
\begin{equation}\label{momentum}
    \dfrac{\mu}{k}\textbf{v}=-\nabla p-g\rho_f(T^f)\textbf{k}+\widetilde{\mu}\Delta \textbf{v},
\end{equation}
where $\textbf{v}$, $p$ and $\rho_f$ are the seepage velocity, pressure and fluid density respectively, $\textbf{k} = (0, 0, 1)$, $g$ is the modulus of gravity acceleration,
 $k$ is the permeability of the porous medium, $\mu$ is the fluid dynamic viscosity, while $\widetilde{\mu}$ is the effective viscosity $(\mu\neq\widetilde\mu)$.

 Along with the continuity equation and the energy balance equations for both fluid and solid phases, Eq. \eqref{momentum} provides the following model
\begin{equation} \label{mod}
\begin{cases}
    \dfrac{\mu}{k}\textbf{v}=-\nabla \widetilde{p}-2g\rho_0\alpha T_0T^f\textbf{k}+g\rho_0\alpha (T^f)^2\textbf{k}+\widetilde{\mu}\Delta \textbf{v},\\
    \nabla\cdot \textbf{v}=0,\\
    \varepsilon(\rho c)_f\dfrac{\partial T^f}{\partial t}+(\rho c)_f\textbf{v}\cdot \nabla T^f=\varepsilon k_f\Delta T^f+h(T^s-T^f),\\
    (1-\varepsilon)(\rho c)_s \dfrac{\partial T^s}{\partial t}=(1-\varepsilon) k_s\Delta T^s-h(T^s-T^f),
\end{cases}
\end{equation}
where $\widetilde{p}$ is defined as follows
\begin{equation}
\widetilde{p}=p+g\rho_0(1-\alpha T_0^2)z,
\end{equation}
and $\varepsilon$, $\rho_i$, $c_i$ and $k_i$ ($i=f,s$) are medium porosity, density, specific heat and thermal conductivity, respectively, of fluid and solid phases. In the present framework, the boundary conditions \eqref{bc_iniz} become the following
\begin{equation}\label{bc1}
\begin{split}
    &T^s=T^f=T_L=0^\circ C,\quad z=0   \\
    &T^s=T^f=T_U\geq 4^\circ C, \quad z=d\\ 
    &\textbf{v}\cdot \textbf{n}=0 ,\quad z=0,d
\end{split}
\end{equation}
where Eq. \eqref{bc1}$_3$ models the impervious planes, being $\textbf{n}$ the outward unit normal to planes $z=0, d$. 

In system \eqref{mod} the Oberbeck-Boussinesq approximation is considered. According to this hypothesis, the quadratic density law \eqref{quadraticdensity} has been retained only in the gravitational term in Eq. \eqref{mod}$_1$ and set constant elsewhere. The previous hypothesis implies that the fluid can be considered \emph{as if} incompressible. Let us remark that, in non-isothermal processes it is not possible to assume the fluid as \emph{incompressible} at all. The incompressibility assumption requires that all the constitutive functions do not depend on pressure. Accounting for this hypothesis, M\"uller \cite{Muller1985} proved that the only function $\rho=\rho(T)$ compatible with the entropy principle is a constant. This conclusion is actually in disagreement with empirical observations and, in particular, with the Oberbeck-Boussinesq approximation. In order to fix this contradiction, Gouin et al. in \cite{gouinmuller2012}, introduced the class of
quasi-thermal-incompressible fluids, i.e. media for
which the fluid density is \emph{the only} function independent
on the pressure among all the constitutive equations. The Authors proved that a quasi-thermal-incompressible fluid tends to be perfectly incompressible when $p\ll p_{cr}$, being $ p_{cr}$ a critical pressure value below which the incompressibility assumption is thermodynamically consistent. In this paper, we will assume $p\ll \widetilde{p}_{cr}$, where the critical pressure value $\widetilde{p}_{cr}$ for the validity of the constitutive density function \eqref{quadraticdensity} can be found in \cite{giuseppe1}. Let us remark that the thermodynamic consistency issues, in the context of thermal convection problems in both clear fluids and porous media, attracted
considerable attention in recent years \cite{passerini2014benard,corli2019benard,arnone2023compressibility}. \\

Model \eqref{mod} is now non-dimensionalised by introducing the following set of variables, where the asterisks denote nondimensional fields
\begin{equation}\label{nondim1}
    \begin{split}
 &\textbf{x}=d\textbf{x}^*,\quad t=\tau t^*\\
 &\textbf{v}=V\textbf{v}^*,\quad \widetilde{p}=P p^*,\quad T^f=T^{\#}(T^f)^*,\quad T^s=T^{\#}(T^s)^*,
    \end{split}
\end{equation}
where
\begin{equation}\label{nondim2}
\begin{split}
  &\tau=\dfrac{(\rho c)_f d^2}{k_f}, \\
  &V=\dfrac{\varepsilon k_f}{(\rho c)_fd},\quad  P=\dfrac{\mu\varepsilon k_f}{k(\rho c)_f} ,\quad T^{\#}=\sqrt{\dfrac{\mu\varepsilon k_f}{2g\rho_0\alpha k (\rho c)_f d}} .
\end{split}
\end{equation}

By substituting \eqref{nondim1}-\eqref{nondim2} into  \eqref{mod}, we end up with the following nondimensional model (where asterisks have been dropped out of notation's convenience):
\begin{equation}\label{modad}
    \begin{cases}
        \textbf{v}=-\nabla p -\Ra \zeta T^f\textbf{k}+\dfrac{(T^f)^2}{2}\textbf{k}+\Da \Delta \textbf{v},\\
        \nabla\cdot \textbf{v}=0,\\
        \dfrac{\partial T^f}{\partial t}+\textbf{v}\cdot \nabla T^f=\Delta T^f +H(T^s-T^f),\vspace{1mm}\\
        A\dfrac{\partial T^s}{\partial t}=\Delta T^s-H\gamma(T^s-T^f),
    \end{cases}
\end{equation}
being 
\begin{equation}
    H=\dfrac{h d^2}{\varepsilon k_f},\quad A=\dfrac{(\rho c)_s k_f}{(\rho c)_f k_s},\quad \gamma=\dfrac{\varepsilon k_f}{(1-\varepsilon )k_s},
\end{equation}
interaction heat transfer coefficient, diffusivity ratio, weighted conductivity ratio, respectively. While the non-dimensional parameters $\Ra$ (the thermal Darcy-Rayleigh
number), $\Da$ (the Darcy number) and $\zeta$ are respectively given by
\begin{equation}\label{def_Ra}
 \Ra=\sqrt{\dfrac{2g\rho_0\alpha k (\rho c)_f d T_U^2}{\mu\varepsilon k_f}},\quad \Da=\dfrac{\widetilde{\mu}k}{d^2\mu},\quad \zeta=\dfrac{4}{T_U}.
\end{equation}
It is worth remarking that, by definition, 
\begin{equation}
    \Ra=\dfrac{T_U}{T^{\#}},
\end{equation}
therefore, boundary conditions \eqref{bc1} become
\begin{equation}\label{bc2}
\begin{split}
& T^s=T^f=0,\quad z=0,\\
& T^s=T^f=\Ra,\quad z=1,\\
& \textbf{v}\cdot \textbf{n}=0,\quad z=0,1.
\end{split}
\end{equation}

Model \eqref{modad} and \eqref{bc2} admits the following steady solution as basic motion (conduction solution), according to which fluid is at rest and heat spreads by conduction:
\begin{equation}\label{mb}
    m_b=(\textbf{v}_b,p_b,T^f_b,T^s_b)
\end{equation}
with
\begin{equation}
\begin{split}
   &\textbf{v}_b=(0,0,0),  \\
   &p_b=cost.-g\rho_0 d P^{-1}(1-\alpha T_0^2)z-\dfrac{\Ra \ \zeta }{2}z^2 +\dfrac{\Ra^2}{6}z^3,\\
   &T^f_b=\Ra \ z,\\
   &T^s_b=\Ra \ z.
\end{split}
\end{equation}

In order to study the stability of $m_b$, let us introduce the perturbation fields $(\textbf{u},\pi,\theta,\phi)$ to velocity, pressure,
fluid temperature and solid temperature, respectively. The following solution of \eqref{modad} originates once perturbations on initial data are applied:
\begin{equation}\label{solnew}
    \textbf{v}=\textbf{v}_b+\textbf{u},\quad p=p_b+\pi,\quad T^f=T^f_b+\theta,\quad T^s=T^s_b+\phi.
\end{equation}
Substituting \eqref{solnew} into \eqref{modad}, we get
\begin{equation}\label{modad2}
    \begin{cases}
        \textbf{u}=-\nabla \pi -\Ra M(z)\theta\textbf{k}+\dfrac{\theta^2}{2}\textbf{k}+\Da \Delta \textbf{u},\\
        \nabla\cdot \textbf{u}=0,\\
        \dfrac{\partial \theta}{\partial t}+\textbf{u}\cdot \nabla\theta+\Ra w=\Delta \theta +H(\phi-\theta),\vspace{1mm}\\
        A\dfrac{\partial \phi}{\partial t}=\Delta \phi-H\gamma(\phi-\theta),
    \end{cases}
\end{equation}
with $M(z)=\zeta-z$ and with the \emph{stress-free} boundary conditions
\begin{equation}\label{bcNL}
   u_z=v_z= w=\phi=\theta=0,\quad z=0,1,
\end{equation}
and initial conditions
\begin{equation}\label{IC}
 \textbf{u}(\textbf{x},0) \! = \!\textbf{u}_0(\textbf{x}) \,,\quad \theta(\textbf{x},0)\!=\!\theta_0(\textbf{x})\,,\quad \phi(\textbf{x},0)\!=\!\phi_0(\textbf{x})\,,\quad \pi(\textbf{x},0)\!=\! \pi_0(\textbf{x}),
\end{equation} 
where $\nabla\cdot\textbf{u}_0=0$.

In the following, motivated by the physics of the problem, we assume that perturbations are periodic in the $x$ and $y$ direction with periods $\frac{2\pi}{k_x}$ and $\frac{2\pi}{k_y}$, respectively. Moreover, $\forall g\in\{\textbf{u}, \pi, \theta, \phi\}$
\begin{equation}
    g \ : (\textbf{x},t) \in \Omega\times \R^+ \rightarrow g(\textbf{x},t)\in\R \ \text{and } g\in W^{2,2}(\Omega) \ \forall t\in\R^+,
\end{equation}
where we denote by $\Omega$ the periodicity cell
\begin{equation}
    \Omega = \left[0,\frac{2\pi}{k_x}\right]\times\left[0,\frac{2\pi}{k_y}\right]\times[0,1]
\end{equation}
and $g$ can be expanded in a Fourier series uniformly convergent in $\Omega$.

\section{Instability of $m_b$}\label{linear}
In this section, we are going to study the instability of the basic motion $m_b$, with the aim of determining the critical Rayleigh number beyond which thermal instability occurs. Hence, by neglecting nonlinear terms in \eqref{modad2}, we get
\begin{equation}\label{modad3}
    \begin{cases}
        \textbf{u}=-\nabla \pi -\Ra M(z)\theta\textbf{k}+\Da \Delta \textbf{u},\\
        \nabla\cdot \textbf{u}=0,\\
        \dfrac{\partial \theta}{\partial t}+\Ra w=\Delta \theta +H(\phi-\theta),\vspace{1mm}\\
        A\dfrac{\partial \phi}{\partial t}=\Delta \phi-H\gamma(\phi-\theta),
    \end{cases}
\end{equation}
with boundary conditions
\begin{equation}\label{bclin}
    u_z=v_z=w=\phi=\theta=0,\quad z=0,1.
\end{equation}
System \eqref{modad3} is autonomous, therefore we can look for solutions such that
\begin{equation}\label{auton}
    f(\textbf{x},t)=\widehat{f}(\textbf{x})e^{\sigma t} \quad \forall f\in \{\textbf{u},\pi,\theta,\phi\},
\end{equation}
where $\sigma \in \C$ is the temporal growth rate of perturbation.

Let us now apply the double curl to Eq. \eqref{modad3}$_1$ and retain only the third component of the resulting equation. 
Substitution of \eqref{auton} into the resulting model originating from \eqref{modad3} will lead to
\begin{equation}\label{modad4}
    \begin{cases}
       \Delta \widehat{w}=-\Ra M(z)\Delta_1\widehat{\theta}+\Da\Delta\Delta \widehat{w}, \\
        \sigma \widehat{\theta}+\Ra \widehat{w}=\Delta \widehat{\theta} +H(\widehat{\phi}-\widehat{\theta}),\\
        \sigma \dfrac{A}{\gamma} \widehat{\phi}=\dfrac{1}{\gamma}\Delta \widehat{\phi}-H(\widehat{\phi}-\widehat{\theta}).
    \end{cases}
\end{equation}
By defining the operator $\mathcal{L}$ as follows
\begin{equation}
    \mathcal{L}:=\Delta-\Da \Delta\Delta,
\end{equation}
and by applying this operator to \eqref{modad4}$_2$ and \eqref{modad4}$_3$, we get
\begin{equation}
    \begin{cases}
       \mathcal{L} \widehat{w}=-\Ra M(z)\Delta_1\widehat{\theta},\\
        \sigma \mathcal{L}\widehat{\theta}+\Ra \mathcal{L}\widehat{w}=\mathcal{L}\Delta \widehat{\theta} +H(\mathcal{L}\widehat{\phi}-\mathcal{L}\widehat{\theta}),\\
        \sigma \dfrac{A}{\gamma} \mathcal{L}\widehat{\phi}=\dfrac{1}{\gamma}\mathcal{L}\Delta \widehat{\phi}-H(\mathcal{L}\widehat{\phi}-\mathcal{L}\widehat{\theta}),
    \end{cases}
\end{equation}
from which, easily,
\begin{equation}\label{eq1}
    \begin{cases}
        \sigma \mathcal{L}\widehat{\theta}-\Ra^2 M(z)\Delta_1\widehat{\theta}=\mathcal{L}\Delta \widehat{\theta} +H(\mathcal{L}\widehat{\phi}-\mathcal{L}\widehat{\theta}),\\
        \sigma \dfrac{A}{\gamma} \mathcal{L}\widehat{\phi}=\dfrac{1}{\gamma}\mathcal{L}\Delta \widehat{\phi}-H(\mathcal{L}\widehat{\phi}-\mathcal{L}\widehat{\theta}).
    \end{cases}
\end{equation}

Let us multiply \eqref{eq1}$_1$ by $\theta^*$ and \eqref{eq1}$_2$ by $\phi^*$, where the asterisk denotes the complex conjugate and integrate over the periodicity cell $\Omega$. Denoting by $\scal{\cdot}{\cdot}$ and by $\|\cdot\|$ the standard scalar product and the associated norm on the Hilbert space $L^2(\Omega)$, we obtain
\begin{equation}\label{eq2}
    \begin{cases}
        \sigma \scal{\mathcal{L}\widehat{\theta}}{\widehat{\theta}^*}=\Ra^2\scal{M(z)\Delta_1\widehat{\theta}}{\widehat{\theta}^*}+\scal{\mathcal{L}\Delta\widehat{\theta}}{\widehat{\theta}^*}+H\left(\scal{\mathcal{L}\widehat{\phi}}{\widehat{\theta}^*}-\scal{\mathcal{L}\widehat{\theta}}{\widehat{\theta}^*}\right),\vspace{1mm}\\
        \sigma \dfrac{A}{\gamma}\scal{\mathcal{L}\widehat{\phi}}{\widehat{\phi}^*}=\dfrac{1}{\gamma}\scal{\mathcal{L}\Delta\widehat{\phi}}{\widehat{\phi}^*}-H\left(\scal{\mathcal{L}\widehat{\phi}}{\widehat{\phi}^*}-\scal{\mathcal{L}\widehat{\theta}}{\widehat{\phi}^*}\right).
    \end{cases}
\end{equation}
By virtue of periodicity assumption and boundary conditions \eqref{bclin}, we get
\begin{equation}
    \begin{split}
        \sigma\left(-\|\nabla\widehat{\theta}\|^2-\Da\|\nabla\nabla\widehat{\theta}\|^2\right)=&-\Ra^2\scal{M(z)}{|\nabla_1\widehat{\theta}|^2}-\|\nabla\nabla\widehat{\theta}\|^2-\Da\|\nabla\nabla\nabla\widehat{\theta}\|^2\\
        &+H\left(\scal{\mathcal{L}\widehat{\phi}}{\widehat{\theta}^*}-\|\nabla\widehat{\theta}\|^2-\Da\|\nabla\nabla\widehat{\theta}\|^2\right),\\
        \\
        \sigma\left(-\dfrac{A}{\gamma}\|\nabla\widehat{\phi}\|^2-\dfrac{A\Da}{\gamma}\|\nabla\nabla\widehat{\phi}\|^2\right)=&-\dfrac{1}{\gamma}\|\nabla\nabla\widehat{\phi}\|^2-\dfrac{\Da}{\gamma}\|\nabla\nabla\nabla\widehat{\theta}\|^2\\
        &-H\left(\|\nabla\widehat{\phi}\|^2-\Da\|\nabla\nabla\widehat{\phi}\|^2-\scal{\mathcal{L}\widehat{\theta}}{\widehat{\phi}^*}\right).
    \end{split}
\end{equation}
Adding the two previous equations, it turns out that
\begin{equation}
    \sigma\in\mathbb{R}.
\end{equation}
This is immediate since the following chain of equivalences holds
\begin{equation}
    \scal{\mathcal{L}\widehat{\phi}}{\widehat{\theta}^*}=\scal{\widehat{\phi}}{\mathcal{L}(\widehat{\theta}^*)}=\scal{\widehat{\phi}}{(\mathcal{L}\widehat{\theta})^*}=\scal{\mathcal{L}\widehat{\theta}}{\widehat{\phi}^*}^*,
\end{equation}
therefore,
\begin{equation}
    \scal{\mathcal{L}\widehat{\theta}}{\widehat{\phi}^*}+\scal{\mathcal{L}\widehat{\theta}}{\widehat{\phi}^*}^*\in\mathbb{R}.
\end{equation}

As a consequence, we can claim the validity of the \emph{strong form of the principle of exchange of stabilities}. This implies that, provided that conditions ensuring instability are verified, perturbations on the basic motion grow without oscillations \cite{landau2013fluid}. Consequently, a new stable steady configuration is reached. In this context, it is said that thermal convection occurs \emph{via steady motions}.

Turning our attention back to system \eqref{modad4}, because of periodicity of perturbation fields, we can write solutions as
\begin{equation}\label{sol}
    \widehat{f}(x,y,z)= \sum_{n=1}^{+\infty} \widetilde{F}_n(x,y,z) \qquad \widehat{f}\in \{\textbf{u},\theta,\phi\},
\end{equation}
where
\begin{equation}\label{sol2}
    \Delta_1\widetilde{F}_n(x,y,z)=-k^2\widetilde{F}_n(z),\qquad k^2=k_x^2+k_y^2.
\end{equation}

Once \eqref{sol}-\eqref{sol2} is plugged into \eqref{modad4}, linearity of the model allows to retain only the $n$-th component. Hence, dropping the subscript, system \eqref{modad4} becomes
\begin{equation}\label{modad5}
    \begin{cases}
       (D^2-a^2) \widetilde{W}=\Ra M(z)k^2\widetilde{\Theta}+\Da (D^2-k^2)^2\widetilde{W} ,\\
        \sigma \widetilde{\Theta}+\Ra \widetilde{W}=(D^2-k^2) \widetilde{\Theta} +H(\widetilde{\Phi}-\widetilde{\Theta}),\\
        \sigma A \widetilde{\Phi}=(D^2-k^2) \widetilde{\Phi}-H\gamma(\widetilde{\Phi}-\widetilde{\Theta}),
    \end{cases}
\end{equation}
where we denote by $D^2=\frac{d^2}{dz^2}$, with boundary conditions
\begin{equation}\label{bc5}
   D^2\widetilde{W}=\widetilde{W}=\widetilde{\Theta}=\widetilde{\Phi}=0,\quad z=0,1.
\end{equation}
System \eqref{modad5}-\eqref{bc5} is a differential eigenvalue problem of this kind
\begin{equation}
    \mathcal{A}\textbf{X}=\sigma\mathcal{B}\textbf{X} \qquad \textbf{X}=(\widetilde{W},\widetilde{\Theta},\widetilde{\Phi}).
\end{equation}
The presence of $z$-dependent coefficients makes the problem demanding from the analytical view point, therefore we implement and employ the Chebyshev-$\tau$ method. The idea behind this numerical procedure involves the discretisation of differential operators $\mathcal{A}$ and $\mathcal{B}$ by mean of Chebyshev polynomials, taking advantage of their several good properties. Once the problem has been reduced to an algebraic eigenvalue problem, the common MatLab routine \texttt{eig} is employed. Further details on the method are provided in Appendix \ref{sec_numerical}.

It is worth remarking that, having proved the strong form of the principle of exchange of stabilities, we expect the numerical method to provide only real values for $\sigma$. This check has been undertaken and verified, proving the goodness of numerical results. 

In this framework, it is well known that $\sigma=0$ allows us to determine the neutral stability curve, which delimits the instability region. Therefore, we are interested in determining first those couples $(k,\Ra)$ such that $\sigma(k,\Ra)=0$ and then, among them, finding the couple $(k_c, \Ra_L)$ that solves the following minimum problem 
\begin{equation}
   \min_{k^2\in\mathbb{R}^+}\Ra^2.
\end{equation}

\section{Nonlinear stability analysis}
In this section we are going to study the nonlinear stability of the conduction solution $m_b$ \eqref{mb} via the well-established energy method \cite{energymethod2004, GaldiRionero1985}. 
Let us introduce the following weighted energy functional
\begin{equation}\label{energy}
    E(t)=\dfrac{1}{2}\scal{g(z)}{\theta^2}+\dfrac{A}{2\gamma}\scal{g(z)}{\phi^2},
\end{equation}
where $g(z)$ is a positive real function to be suitably chosen and $\scal{\cdot}{\cdot}$ and $\|\cdot\|$ are the real scalar product on $L^2(\Omega)$ and the related norm, respectively. If Eq. \eqref{modad2}$_1$ is multiplied by $\textbf{u}$ and integrated over $\Omega$,   Eq. \eqref{modad2}$_3$ is multiplied by $g(z)\theta$ and integrated over $\Omega$, Eq. \eqref{modad2}$_4$ is multiplied by $g(z)\phi$ and integrated over $\Omega$, the sum of the resulting equations can be written as
\begin{equation}\label{eqenergy}
\begin{split}
  \dfrac{dE}{dt}=&\,\dfrac{1}{2}\scal{g'(z)w}{\theta^2}-\Ra\scal{g(z)w}{\theta}-\scal{\theta_z}{\theta g'(z)}-\scal{g(z)}{|\nabla\theta|^2}\\
  &+2H\scal{g(z)\phi}{\theta}-H\scal{g(z)}{\theta^2}-\dfrac{1}{\gamma}\scal{\phi_z}{\phi g'(z)}-\dfrac{1}{\gamma}\scal{g(z)}{|\nabla\phi|^2}\\
  &-H\scal{g(z)}{\phi^2}+H\scal{g(z)\phi}{\theta}-\|\textbf{u}\|^2-\Ra\scal{M(z)\theta}{w}\\
  &-\dfrac{1}{2}\scal{w}{\theta^2}-\Da\|\nabla\textbf{u}\|^2.
\end{split}
\end{equation}
Now, in order to handle the cubic term $\scal{g'(z)w}{\theta^2}-\scal{w}{\theta^2}$, we choose
\begin{equation}
    g(z)=\mu-z, \qquad \mu>1,
\end{equation}
where $\mu$ is a parameter to be optimally chosen later. In such a way, \eqref{eqenergy} becomes
\begin{equation}\label{en1}
    \dfrac{dE}{dt}=\Ra I-D,
\end{equation}
where
\begin{equation}\label{dissipation}
    \begin{split}
        I=&\, -\scal{(\mu+\zeta-2z)w}{\theta},\\
        D=&\, \|\textbf{u}\|^2+\Da\|\nabla\textbf{u}\|^2+\scal{\mu-z}{|\nabla\theta|^2}\\
        &\qquad\quad\qquad+\dfrac{1}{\gamma}\scal{\mu-z}{|\nabla\phi|^2}+H\scal{\mu-z}{|\theta-\phi|^2}
    \end{split}
\end{equation}
are respectively the production term and the dissipation term.

Hence, by denoting
\begin{equation}\label{varprob}
    \dfrac{1}{\Ra_E}=\max_{\mathcal{H}}\dfrac{I}{D},
\end{equation}
where the space of kinematically admissible functions is defined as:
\begin{equation}
\begin{split}
  \mathcal{H}=\Bigl\{(\textbf{u},\theta,\phi)\in W^{1,2}(\Omega)\;&|\;\nabla \cdot \textbf{u}=0,\;x,y\;\text{periodic}\\
  &\text{with period }2\pi/k_x,2\pi/k_y,\;\text{verifying (\ref{bcNL})}\Bigr\} ,  
\end{split}
\end{equation}
from \eqref{en1}, we obtain
\begin{equation}\label{en2}
    \dfrac{dE}{dt}=\Ra I-D=-D\left(1-\Ra \dfrac{I}{D}\right)\leq -\left(\dfrac{\Ra_E-\Ra}{\Ra_E}\right)D.
\end{equation}
Looking at definition \eqref{dissipation}$_2$, by applying the weighted Poincaré inequality for which
\begin{equation}
    \scal{\mu-z}{\theta^2}\leq c_P\scal{\mu-z}{|\nabla\theta|^2},\quad \scal{\mu-z}{\phi^2}\leq c_P\scal{\mu-z}{|\nabla\phi|^2},
\end{equation}
we get from \eqref{en2}
\begin{equation}
    \dfrac{dE}{dt} \leq -c\left(\dfrac{\Ra_E-\Ra}{\Ra_E}\right)E,
\end{equation}
where
\begin{equation}
    c=\max\left\{\dfrac{2}{Ac_P},\frac{2}{c_P}\right\}.
\end{equation}
Hence, if $\Ra<\Ra_E$,
\begin{equation}
    E(t)\leq E(0)e^{-\alpha t}
\end{equation}
where
\begin{equation}
    \alpha=c\dfrac{\Ra_E-\Ra}{\Ra_E}.
\end{equation}
We have recovered the exponential decay in time of the energy $E(t)$, Eq. \eqref{energy}, when $\Ra<\Ra_E$. As a consequence, as long as  $\Ra<\Ra_E$, perturbations on fluid and solid temperature tend to zero exponentially as $t\rightarrow +\infty$. In order to discuss the nonlinear stability of $m_b$ within the regime for $\Ra<\Ra_E$, we need to determine the faith of the seepage velocity norm. We are able to show that we can control $\|\textbf{u}\|$ with $\|\theta\|$.

Indeed, let us multiply \eqref{modad2}$_1$ by $\textbf{u}$ and integrate over $\Omega$. The resulting equation will be
\begin{equation}
    \|\textbf{u}\|^2 + \Da \|\nabla\textbf{u}\|^2 = |\Ra\scal{g(z) \theta}{w}| + \left|\frac{1}{2} \scal{\theta^2}{w}\right|.
\end{equation}
The Cauchy-Schwartz inequality and the Poincaré inequality lead to
\begin{equation}
    \left(\dfrac{1}{4}+\dfrac{\Da}{c_p}\right)\|\textbf{u}\|^2\leq \dfrac{M^2 \Ra^2}{2}\|\theta\|^2+\dfrac{1}{4}\|\theta^2\|^2,
\end{equation}
where $ M:=\max_{z\in [0,1]}|\zeta-z| $.
The following Lemma comes to help on estimating and controlling $\|\theta^2\|^2$:

\begin{lem}\label{lemma}
	Let $\Omega_1$ and $\Omega_2$ be sets that partition the periodicity cell $\Omega$ such that
	\begin{equation}\label{}
		\begin{split}
		&\Omega_1=\left\{\textbf{x}\in\Omega\;:\; \theta(\textbf{x},t)>\Ra-T_b\right\},\\
		&\Omega_2=\left\{\textbf{x}\in\Omega\;:\; \theta(\textbf{x},t)\leq\Ra-T_b\right\}	.
		\end{split}
	\end{equation}
If 
\begin{equation}\label{}
	\theta_0(\textbf{x}),\phi_0(\textbf{x})\in W^{2,2}(\Omega)\cap L^\infty(\Omega),
\end{equation}
then, there exists a positive constant $ \Gamma $ such that
\begin{equation}\label{}
	\theta(\textbf{x},t)+T_b(z)-\Ra\leq \Gamma,
\end{equation}
with 
\begin{equation}\label{}
	\Gamma=\begin{cases}
		\bar{\theta}_0 &\quad \text{if}\quad \phi_0\leq \Ra -T_b(z), \\
		\bar{\theta}_0+\bar{\phi}_0&\quad \text{otherwise},
	\end{cases} 
\end{equation}
and
\begin{equation}\label{}
	\begin{split}
		&\bar{\theta}_0=\underset{\Omega_1}{\esssup}\left\{\left(\theta_0(\textbf{x})+T_b(z)-\Ra\right)_+\right\},\\
		&\bar{\phi}_0=\underset{\Omega}{\esssup}\left\{\left(\phi_0(\textbf{x})+T_b(z)-\Ra\right)_+\right\}.
	\end{split}
\end{equation}
\end{lem}
\begin{proof}
   Since the proof is lenghtly, it shall be delegated to Appendix \ref{proof}.
\end{proof}

Consequently, recalling the inequality $(a+b)^2\leq 2(a^2+b^2)$ and applying Lemma \ref{lemma}, one obtains
\begin{equation}
\begin{split}
\|\theta^2\|^2&\leq 2\int_{\Omega_1}(\theta+T_b-\Ra)^2\theta^2d\Omega+2\int_\Omega(-T_b+\Ra)^2\theta^2d\Omega\\
&\leq \Upsilon \|\theta\|^2 
\end{split}
\end{equation}
where
\begin{equation} 
\Upsilon=\max\Bigl\{4\Bigl[ \bar{\theta}_0+\bar{\phi}_0\Bigr]^2,2\max_{z\in [0,1]}\left[(-T_b+\Ra)^2\right] \Bigr\}   
\end{equation}

and 
\begin{equation}
    \left(\dfrac{1}{4}+\dfrac{\Da}{c_p}\right)\|\textbf{u}\|^2\leq \dfrac{M^2 \Ra^2}{2}\|\theta\|^2+\dfrac{\Upsilon}{4}\|\theta\|^2.
\end{equation}

In conclusion, we recover the following stability theorem
\begin{thm}
    If $\Ra<\Ra_E$, then the conduction solution $m_b$ is globally nonlinearly stable.
\end{thm}

The critical threshold $\Ra_E$ is determined by solving the variational problem \eqref{varprob}. The Euler-Lagrange equations are
\begin{equation}\label{E-L}
    \begin{cases}
        \textbf{u}-\Da \Delta\textbf{u}+\Ra_E F(z)\theta\textbf{k}=\nabla l,\\
        \Ra_E F(z)w+\theta_z-g(z)\Delta \theta+Hg(z)(\theta-\phi)=0,\\
        \phi_z-g(z)\Delta\phi-H\gamma g(z)(\theta-\phi)=0,
    \end{cases}
\end{equation}
where $F(z)=\frac{\mu}{2}-\frac{\zeta}{2}-z$ and $l$ is a Lagrange multiplier.

By retaining the third component of the double curl of \eqref{E-L}$_1$ and employing \eqref{sol}-\eqref{sol2} to write the solution, we get
\begin{equation}\label{EL2}
    \begin{cases}
    (D^2-k^2) \widetilde{W} -\widetilde{Y}=0,\\
        -\widetilde{Y}+\Da (D^2-k^2) \widetilde{Y}+\Ra_E F(z)k^2\widetilde{\Theta}=0,\\
        \Ra_E F(z)\widetilde{W}+D\widetilde{\Theta}-g(z)(D^2-k^2) \widetilde{\Theta}+Hg(z)(\widetilde{\Theta}-\widetilde{\Phi})=0,\\
        D\widetilde{\Phi}-g(z)(D^2-k^2)\widetilde{\Phi}-H\gamma g(z)(\widetilde{\Theta}-\widetilde{\Phi})=0,
    \end{cases}
\end{equation}
with the usual boundary conditions:
\begin{equation}\label{bc_EL2}
    \widetilde{Y}=\widetilde{W}=\widetilde{\Theta}=\widetilde{\Phi}=0,\quad z=0,1.
\end{equation}
In order to solve the 8th order differential eigenvalue problem \eqref{EL2}-\eqref{bc_EL2}, with non-constant coefficients, we implement a shooting method coupled with a Newton-Raphson scheme. For additional details on the method we refer to \cite{rees2011, capone2023weakly, rees2000onset}.

To ensure that Eqs. \eqref{EL2}-\eqref{bc_EL2} provide nonzero solution, we need to add one extra boundary condition:
\begin{equation}\label{extrabc}
    D\widetilde{\Theta}=1, \ \text{on }  z=0
\end{equation}
and, consequently, one extra equation
\begin{equation}\label{extraeq}
    D\Ra_E=0,
\end{equation}
which enforces that $\Ra_E$ is constant.

Moreover, the critical Rayleigh number for nonlinear stability will be provided by solving the following problem
\begin{equation}\label{maxprob}
    \Ra_{E,c}^2=\max_{\mu>1}\min_{k^2\in\mathbb{R}^+}\Ra_E^2.
\end{equation}
As a result, system \eqref{EL2}, \eqref{bc_EL2}, \eqref{extrabc}, \eqref{extraeq} is replaced by the following 18th order system
\begin{equation}\label{18th}
    \begin{cases}
    (D^2-k^2) \widetilde{W} -\widetilde{Y}=0,\\
        -\widetilde{Y}+\Da (D^2-k^2) \widetilde{Y}+\Ra_E F(z)k^2\widetilde{\Theta}=0,\\
        \Ra_E F(z)\widetilde{W}+D\widetilde{\Theta}-g(z)(D^2-k^2) \widetilde{\Theta}+Hg(z)(\widetilde{\Theta}-\widetilde{\Phi})=0,\\
        D\widetilde{\Phi}-g(z)(D^2-k^2)\widetilde{\Phi}-H\gamma g(z)(\widetilde{\Theta}-\widetilde{\Phi})=0,\\
        D\Ra_E=0,\\
        (D^2-k^2) \widetilde{W}_k -2k\widetilde{W} -\widetilde{Y}_k=0,\\
        -\widetilde{Y}_k+\Da (D^2-k^2) \widetilde{Y}_k-2k\widetilde{Y}+\Ra_E F(z)k^2\widetilde{\Theta}_k+2\Ra_E F(z)k\widetilde{\Theta}=0,\\
        \Ra_E F(z)\widetilde{W}_k+D\widetilde{\Theta}_k-g(z)(D^2-k^2) \widetilde{\Theta}_k+2g(z)k \widetilde{\Theta}+Hg(z)(\widetilde{\Theta}_k-\widetilde{\Phi}_k)=0,\\
        D\widetilde{\Phi}_k-g(z)(D^2-k^2)\widetilde{\Phi}_k+2g(z)k\widetilde{\Phi}-H\gamma g(z)(\widetilde{\Theta}_k-\widetilde{\Phi}_k)=0,\\
        Dk=0,
    \end{cases}
\end{equation}
where the following boundary conditions are added to \eqref{bc_EL2}-\eqref{extrabc}
\begin{equation}\label{extrabc2}
\begin{aligned}
    &     \widetilde{Y}_k=\widetilde{W}_k=\widetilde{\Theta}_k=\widetilde{\Phi}_k=0,\ \text{on } z=0,1,\\
    &      D\widetilde{\Theta}_k=0, \ \text{on } z=0,
\end{aligned}
\end{equation}
and where the subscript $k$ denotes the partial derivative with respect to the wavenumber $k$.
Finally, system \eqref{bc_EL2}, \eqref{extrabc}, \eqref{18th}, \eqref{extrabc2} is solved for each choice of $\mu>1$.

\section{Numerical results and conclusions}
This section is devoted to numerical results obtained from the application of Chebyshev-$\tau$ method and shooting method to differential eigenvalue problems \eqref{modad5}, \eqref{bc5} and \eqref{18th},\eqref{bc_EL2}, \eqref{extrabc}, \eqref{extrabc2}, respectively. We discuss the effect of parameters characterising the studied physical problem on the onset of thermal instability.

First of all, let us remark that, given the difficulty in assigning a precise value to the interaction heat transfer coefficient $H$, we set an interval where $H$ can vary, following the choice of \cite{govender2007effect}. Therefore, we set $H\in (10^{-4}, 10^{4})$. This choice will allow us to have a better insight of the effect of parameters on the critical Rayleigh number $\Ra^2_L$, defined as a function of $\log_{10} H$.

\begin{figure}[h!]
    \centering
    \includegraphics[scale=0.7]{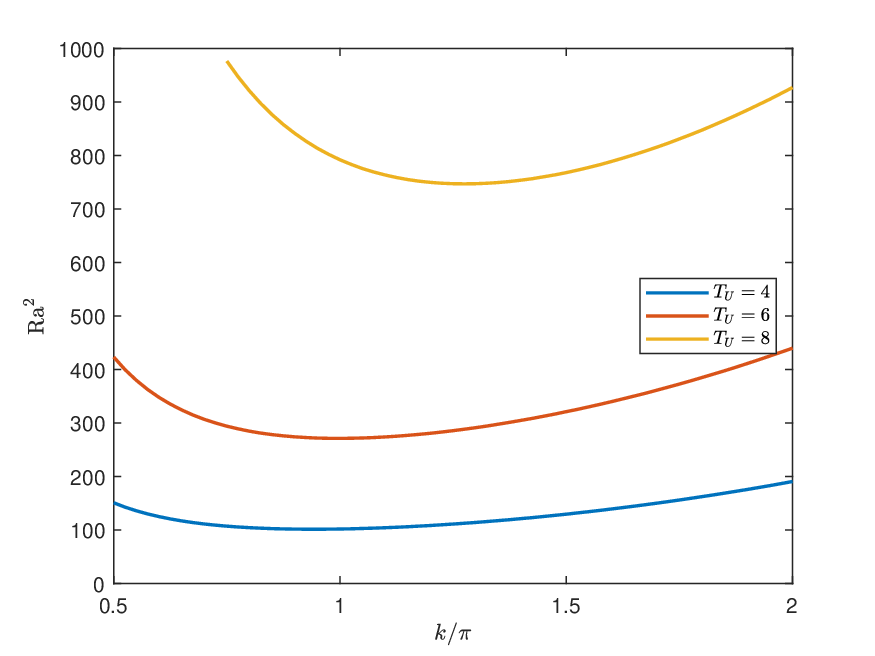}
    \caption{Neutral stability curves for quoted values of the upper plane temperature $T_U$, with $\Da=0.01, \gamma=10$ and $H=100$.}
    \label{fig1}
\end{figure}
In Figure \ref{fig1}, the neutral stability curves for quoted values of the upper plane temperature $T_U$ are depicted. The stabilising effect of $T_U$ on the occurrence of penetrative convection is highlighted. This behaviour is expected from a physical viewpoint given that an increasing upper plane temperature results in a thinning of the potentially unstable fluid region, which is the part of the fluid where its temperature is in the interval $(0,4)^\circ C$. Consequently, this fluid portion struggles to penetrate the upper stable fluid region, resulting in a delay of the onset of penetrative convection.

\begin{figure}[h!]
    \centering
    \includegraphics[scale=0.9]{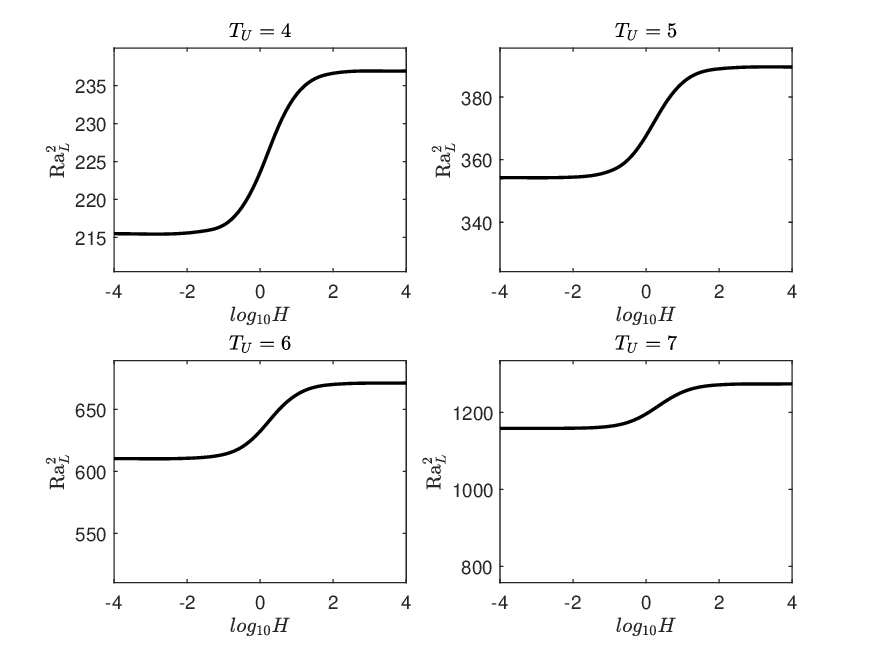}
    \caption{Critical linear Rayleigh number as function of $\log_{10}H$ for quoted values of the upper plane temperature $T_U$, with $\Da=0.1$ and $\gamma=10$.}
    \label{fig2}
\end{figure}
The stabilising effect of $T_U$ is also evident in Figure \ref{fig2}, where, unlike the previous case, $H$ is varying in the aforementioned interval. Hence, Figure \ref{fig2} shows the variation of the critical Rayleigh number from linear analysis with respect to $\log_{10} H$ for four values of $T_U$. It is worth remarking the monotonic behaviour of $\Ra_L$ with $H$, which is a typical trend in problems in the LTNE regime \cite{banu2002onset}. Moreover, let us notice that the critical Rayleigh number tends asymptotically to constant values in the limit for both large and small $H$. Physically, as $H$ approaches $0$ the fluid acts independently from the solid phase as they are not exchanging heat, whereas as $H\rightarrow \infty$, the solid and fluid phases exchange heat so fast that they reach thermal equilibrium immediately and they can be considered as a single phase. In both cases, we recover the LTE regime. As a result, the respective
mathematical problems are identical except for a rescaling of $\Ra_L$. The rescaling factor has been determined analytically in \cite{banu2002onset}. Let us remark that, following definition \eqref{def_Ra}, $\Ra^2$ is the Rayleigh number based on fluid properties, while the rescaled Rayleigh number is based on the porous medium properties:
\begin{equation}
    \frac{\gamma}{\gamma+1} \Ra^2 =\dfrac{2g\rho_0\alpha k (\rho c)_f d T_U^2}{\mu\left[\varepsilon k_f+(1-\varepsilon)k_s\right]}
\end{equation}

\begin{figure}[h!]
    \centering
    \includegraphics[scale=0.9]{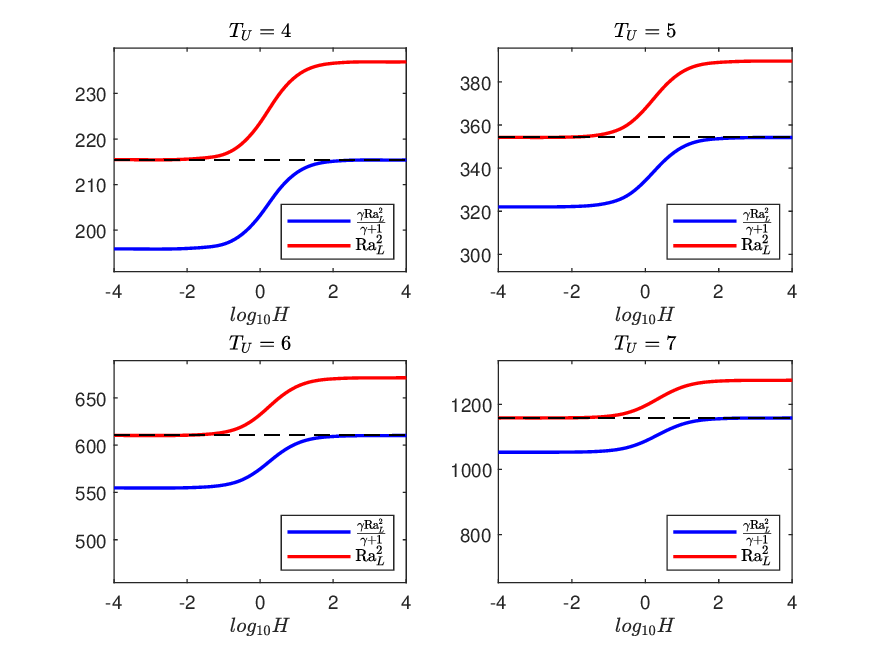}
    \caption{Comparison between critical Rayleigh number and rescaled critical Rayleigh number as functions of $\log_{10}H$ for quoted values of the upper plane temperature $T_U$, with $\Da=0.1$ and $\gamma=10$.}
    \label{fig3}
\end{figure}

 \begin{table}[h!]\centering
 \scriptsize
		\begin{tabular}{@{}l|l|lll@{}}\toprule
	$\Ra_{_L}^2$\cite{giuseppe1}\qquad\qquad\qquad&$\Ra_{_L}^2\;(H=0)$ & $\frac{\gamma\Ra_{_L}^2}{\gamma+1}\;(H=\infty)$  &&$ T_U(^\circ C) $\\ \midrule
77.079 & 77.071 & 77.065 && 4 \\
123.462 & 123.447 & 123.450 && 5 \\
198.030 & 198.009 & 198.026 && 6 \\
313.547 & 313.531 & 313.527 && 7 \\
471.384 & 471.331 & 471.338 && 8 \\
672.119 & 672.072 & 672.050 && 9 \\
921.929 & 921.882 & 921.850 && 10 \\
			\bottomrule
		\end{tabular}
		\caption{Critical Rayleigh numbers in the limit of LTE compared with the ones obtained in \cite{giuseppe1}, with $\Da=0$ and $\gamma=1$, for quoted values of $T_U$.}
		\label{tab1}
	\end{table}

In Figure \ref{fig3} we report a comparison between the critical Rayleigh number and the rescaled critical Rayleigh number. 
As expected, the rescaled Rayleigh number when $H\rightarrow\infty$ approaches $\Ra_L$ in the limit for small $H$ values. 

This agreement is also evident in Table \ref{tab1}, where results obtained for this problem are compared with findings in  \cite{giuseppe1}. Table \ref{tab1} shows a very good agreement, confirming the goodness of the numerical procedure employed in the present paper.

\begin{figure}[h!]
    \centering
    \includegraphics[scale=0.7]{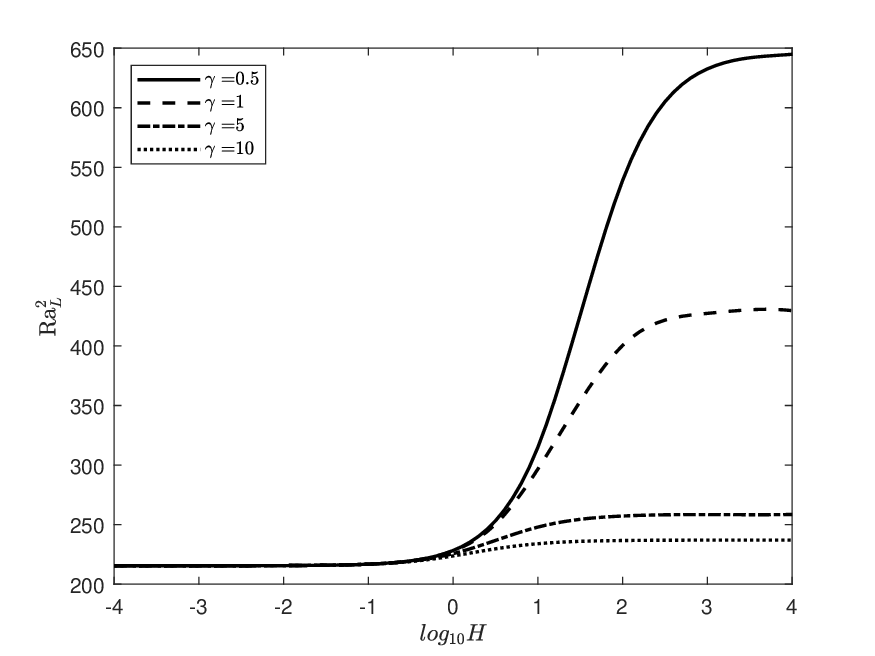}
    \caption{Critical linear Rayleigh number as function of $\log_{10}H$ for quoted values of $\gamma$, with $\Da=0.1$ and $T_U=4$.}
    \label{figgamma}
\end{figure}

\begin{figure}[h!]
    \centering
    \includegraphics[scale=0.7]{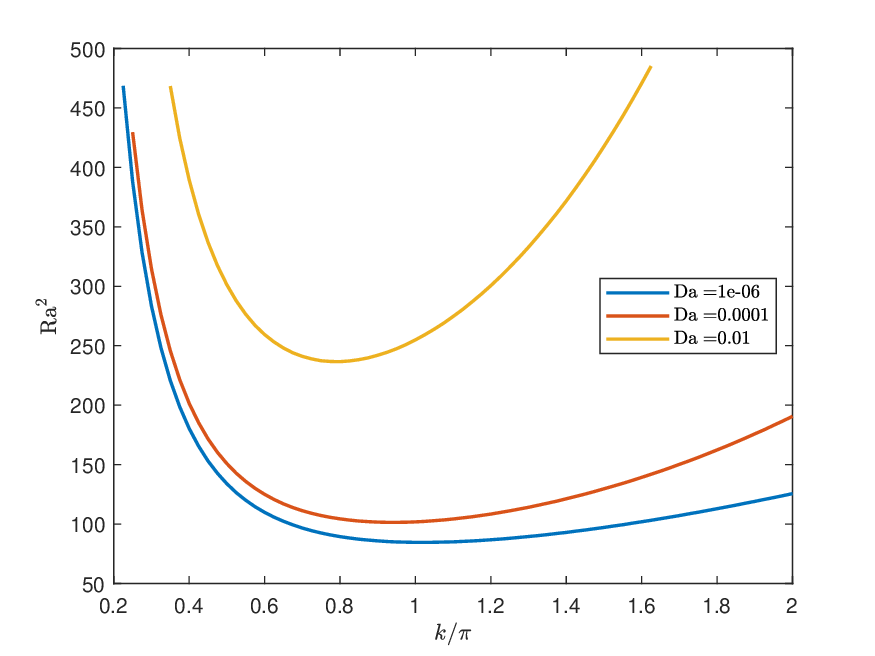}
    \caption{Neutral stability curves for quoted values of the Darcy number $\Da$, with $T_U=4, \gamma=10$ and $H=100$.}
    \label{figDa}
\end{figure}

Figure \ref{figgamma} shows the destabilising effect of $\gamma$ on the onset of penetrative convection. This behaviour is well-known in literature and it is physically reasonable. Large $\gamma$ implies that the weighted fluid thermal conductivity is higher than the solid one. Therefore, heat diffusion is quicker within the fluid phase, resulting in a easier occurrence of thermal instability.

Regarding the effect of high porosity on the onset of instability, Figure \ref{figDa} describes the behaviour of neutral stability curves for prescribed values of the Darcy number. As expected, increasing $\Da$ leads to higher critical Rayleigh number.

 \begin{table}[h!]\centering
 \scriptsize
		\begin{tabular}{@{}ll|llll|llll|lll@{}}\toprule
	\multicolumn{3}{c}{$\zeta=1\, (T_U=4)$} & \phantom{c}&\multicolumn{3}{c}{$\zeta=0.8\, (T_U=5)$} &  \phantom{c}&\multicolumn{3}{c}{$\zeta\simeq 0.6\, (T_U=6)$}&\phantom{c}& \\
			\cmidrule{1-3} \cmidrule{5-7} \cmidrule{9-11}
	$\mu$ &	$\Ra_{E,c}^2$&$ \Ra_{_L}^2 $&& $\mu$ & $\Ra_{E,c}^2$ & $\Ra_{_L}^2 $  && $\mu$ & $\Ra_{E,c}^2$ & $\Ra_{_L}^2$  &&$ \Da $\\ \midrule
		 1.110& 136.089 & 140.583 & & 1.018& 203.914 & 223.555  &&1.004& 282.691 & 352.754 && 0 \\
          1.118&  163.836 & 169.692 & &1.022& 247.388 & 273.790 &&1.005& 347.809 & 447.179 && 0.01 \\
          1.132&  384.144 & 400.464 & &1.031& 584.823 & 657.492  &&1.008& 837.449 & 1102.020 && 0.1 \\
			\bottomrule
		\end{tabular}
		\caption{Comparison between linear and nonlinear critical Rayleigh numbers for quoted values of $T_U$ and $\Da$, with $H=100$ and $\gamma=1$.}
		\label{tab2}
	\end{table}

\begin{figure}[h!]
    \centering
    \includegraphics[scale=0.7]{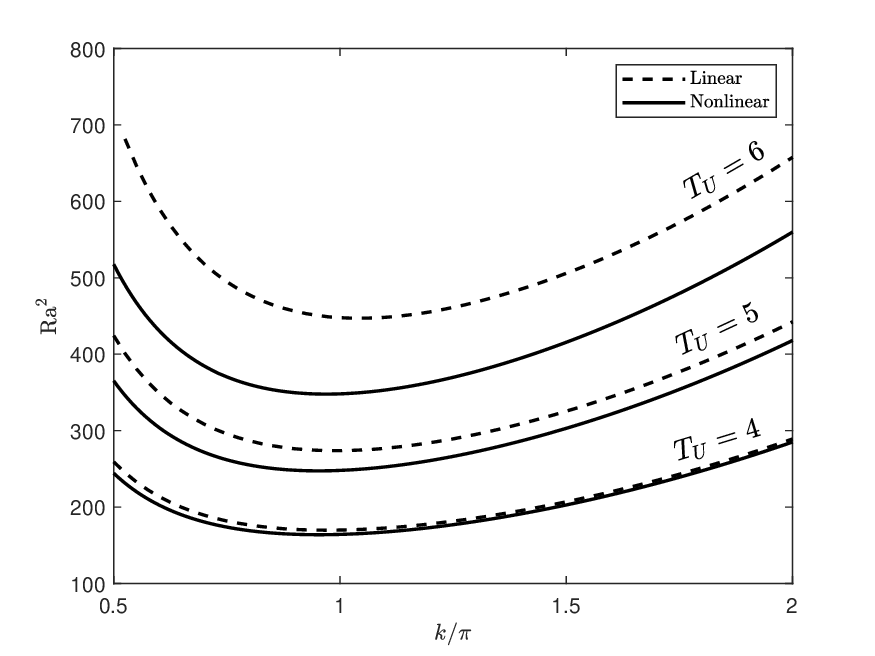}
    \caption{Linear and nonlinear marginal stability curves for quoted values of the upper plane temperature $T_U$ with $\Da=0.01$, $H=100$ and $\gamma=1$.}
    \label{figcamparison}
\end{figure}

Table \ref{tab2} and Figure \ref{figcamparison} are meant to compare results obtained from linear and nonlinear stability analyses. In Figure \ref{figcamparison}, the neutral stability curves from the linear analysis are plotted together with curves obtained solving the Euler-Lagrange equations coming from a weighted energy analysis. \\
\\
It is worth remarking that when the upper temperature is $T_U=4^\circ C$, the stably stratified region vanishes and the entire porous layer becomes potentially unstable. \\
\\
In summary, our results show that
\begin{itemize}
    \item in the limit cases, $H\rightarrow 0 $ and $ H\rightarrow \infty$, i.e. in LTE, the instability thresholds coincide with the ones found in \cite{giuseppe1};
    \item the principle of exchange of stabilities has been rigorously proved;
    \item the upper bounding plane temperature and Darcy's number have both a stabilising effect on the onset of penetrative convection;
    \item in LTNE the instability thresholds show a monotonic trend when $H$ increases;
    \item the nonlinear stability analysis, via weighted energy method has been performed. 
\end{itemize}

 \appendix
 \section{Appendix}
 \subsection{Numerical procedure}\label{sec_numerical}
The eigenvalue problem \eqref{modad5}-\eqref{bc5} was solved thanks to the Chebyshev-$\tau$ method, thoroughly implemented on MatLab, first developed by 2021 Turing-prize-awarded Dongarra (see \cite{dongarra}). In this section, we provide some details useful to the implementation of the method. For a deeper discussion we refer to \cite{arnone2023chebyshev,bourne2003, capone2022natural}.

Let $T_k\in L^2(-1,1)$ be the $k$-th Chebyshev polynomial, with $k\in\N_0$. 
Solution of \eqref{modad5} can be expanded as truncated series of Chebyshev polynomials, i.e.
\begin{equation}\label{expansion2}
 \widetilde{W}= \sum_{k=0}^{N+2} W_k T_k(z), \quad \widetilde{\Theta} = \sum_{k=0}^{N+2} \Theta_k T_k(z), \quad \widetilde{\Phi} = \sum_{k=0}^{N+2} \Phi_k T_k(z),
\end{equation}
paying attention first to transforming the definition interval of the unknown fields from $(0,1)$ to $(-1,1)$ with a suitable transformation \cite{arnone2023onset}.
We can now substitute Eq. \eqref{expansion2} into \eqref{modad5} and then take advantage of orthogonality of Chebyshev polynomials with respect to the scalar product 
\begin{equation}
 \langle f, g\rangle = \int_{-1}^{1} \frac{f g}{\sqrt{1-z^2}} d z \qquad f, g \in L^2(-1,1)
\end{equation}
by taking the inner product with $T_k$ for $k=0,\dots, N$. As a result, we get a generalised eigenvalue problem of this kind
\begin{equation}\label{eigprob}
    \mathcal{A} \textbf{X} = \sigma \mathcal{B} \textbf{X}
\end{equation}
where $\mathcal{A}, \mathcal{B} \in \mathcal{M}^{3(N+1)+6\times 3(N+1)+6}$ are completed by adding six equations derived from boundary conditions. Hence, starting from \eqref{modad5}, we get
\begin{equation}
    \mathcal{A}= \begin{pmatrix}
        4D^2-k^2-\Da\left(4D^2-k^2\right)^2 & \Ra M(z)k^2 & \bm{0}\\
        \text{b.c.}&0\dots0&0\dots0\\
        -\Ra & 4D^2-k^2-H & H\\
        0\dots0&\text{b.c.}&0\dots0\\
        \bm{0} & H\gamma & 4D^2-k^2-H\gamma\\
        0\dots0&0\dots0&\text{b.c.}
    \end{pmatrix}
\end{equation}
where $M(z)=\zeta-\frac{z}{2}-\frac{1}{2}$, $\bm{0}\in\mathcal{M}^{N+1,N+1}$ and
\begin{equation}
    \mathcal{B}= \begin{pmatrix}
        \bm{0} & \bm{0} & \bm{0}\\
        0\dots0&0\dots0&0\dots0\\
        \bm{0} & 1 & \bm{0}\\
        0\dots0&0\dots0&0\dots0\\
        \bm{0} & \bm{0} & A\\
        0\dots0&0\dots0&0\dots0
    \end{pmatrix}
\end{equation}

The main problem when solving \eqref{eigprob} concerns the occurrence of spurious eigenvalues, because of the singularity of $\mathcal{B}$. This is a well recognised problem in literature (see \cite{bourne2003, dongarra,gheorghiu2014spectral,mcfadden}) but a valuable solution is found thanks to boundary conditions. We will not go through this, but we refer the reader to a very detailed explanation in \cite{arnone2023onset}. Once the new eigenvalue problem involves non-singular matrices, eigenvalues can be determined safely.

\subsection{Proof of Lemma \ref{lemma}}\label{proof}

	Let us consider the following system of partial differential equations
	\begin{equation}\label{P}
		\begin{cases}
				\dfrac{\partial T^f}{\partial t}+\textbf{v}\cdot \nabla T^f=\Delta T^f +H(T^s-T^f),\vspace{2mm}\\
				A\dfrac{\partial T^s}{\partial t}=\Delta T^s-H\gamma(T^s-T^f),\\
				\nabla\cdot \textbf{v}=0,
			\end{cases}
	\end{equation}
and the following periodic boundary conditions
\begin{equation}\label{}
	\begin{split}
	&		T^s=T^f=0\quad z=0,\\
	&	T^s=T^f=\Ra\quad z=1,\\
	&\textbf{v}\cdot \textbf{n}=0\quad z=0,1,
\end{split}
\end{equation}
and initial conditions
 \begin{equation}
 \begin{split}
	T^f(\textbf{x},0)&=T^f_0(\textbf{x}),\quad \textbf{x}\in\Omega,\\
	T^s(\textbf{x},0)&=T^s_0(\textbf{x}),\quad \textbf{x}\in\Omega,\\	
	\textbf{v}(\textbf{x},0)&=\textbf{v}_0(\textbf{x}),\quad \textbf{x}\in\Omega,	
	\end{split}
\end{equation}
with $\nabla\cdot \textbf{v}_0=0$. Let $ r\in\{f,s\} $, and let us introduce the following truncated function
\begin{equation}\label{key}
	\left(T^r(\textbf{x},t)-\Ra\right)_+=\begin{cases}
		T^r(\textbf{x},t)-\Ra \quad&\text{if}\quad T^r(\textbf{x},t)>\Ra,\\
		0\quad &\text{otherwise}.
	\end{cases}
\end{equation}
By multiplying \eqref{P}$ _1 $ by $ [(T^f-\Ra)_+]^{2n-1} $ and integrate over the periodicity cell $ \Omega $, we get
\begin{equation}\label{1}
	 \begin{split}
	&\int_\Omega\dfrac{\partial T^f}{\partial t}[(T^f-\Ra)_+]^{2n-1}\,d\Omega+\int_\Omega\textbf{v}\cdot \nabla T^f[(T^f-\Ra)_+]^{2n-1}\,d\Omega\\
	&\qquad\qquad=\int_\Omega\Delta T^f[(T^f-\Ra)_+]^{2n-1}\,d\Omega +H\int_\Omega(T^s-T^f)[(T^f-\Ra)_+]^{2n-1}\,d\Omega.
\end{split} 
\end{equation}
The boundary conditions permit us to write \eqref{1} as
\begin{equation}\label{1.1}
	\begin{split}
		\dfrac{1}{2n}\dfrac{d}{dt}\int_\Omega[(T^f&-\Ra)_+]^{2n}\,d\Omega=-(2n-1)\int_\Omega [(T^f-\Ra)_+]^{2n-2}(\nabla T^f)^2\,d\Omega\\
		&\underline{+H\int_\Omega T^s[(T^f-\Ra)_+]^{2n-1}\,d\Omega-H\int_\Omega T^f[(T^f-\Ra)_+]^{2n-1}\,d\Omega} .
	\end{split} 
\end{equation}
Following the same ideas, if we multiply equation \eqref{P}$ _2 $ by $ [(T^s-\Ra)_+]^{2n-1} $ and integrate over the periodicity cell $ \Omega $, we obtain
\begin{equation}\label{1.2}
	\begin{split}
\dfrac{A}{\gamma}\dfrac{1}{2n}\dfrac{d}{dt}\int_\Omega [(T^s&-\Ra)_+]^{2n}\,d\Omega=-\dfrac{1}{\gamma}(2n-1)\int_\Omega [(T^s-\Ra)_+]^{2n-2}(\nabla T^s)^2\,d\Omega\\
&\underline{-H\int_\Omega T^s[(T^s-\Ra)_+]^{2n-1}\,d\Omega+H\int_\Omega T^f[(T^s-\Ra)_+]^{2n-1}\,d\Omega}.
	\end{split}	
\end{equation}
If we now sum the underlined terms in Eqs. \eqref{1.1} and \eqref{1.2}, we notice, recalling the useful inequality $ (|a|^{p-2}a-|b|^{p-2}b)(a-b)\geq 0 $, that
\begin{equation}
  \begin{split}
	&H\int_\Omega [T^s-T^f]\left([(T^f-\Ra)_+]^{2n-1}-[(T^s-\Ra)_+]^{2n-1}\right)\,d\Omega\\
	&\qquad=-H\int_\Omega[T^f-T^s]\left([(T^f-\Ra)_+]^{2n-1}-[(T^s-\Ra)_+]^{2n-1}\right)\,d\Omega\\
	&\qquad=-H\int_\Omega[(T^f-\Ra)-(T^s-\Ra)]\left([(T^f-\Ra)_+]^{2n-1}-[(T^s-\Ra)_+]^{2n-1}\right)\,d\Omega\\
	&\qquad=-H\int_\Omega[(T^f-\Ra)_+-(T^s-\Ra)_+]\left([(T^f-\Ra)_+]^{2n-1}-[(T^s-\Ra)_+]^{2n-1}\right)\,d\Omega \\
	&\qquad\leq 0,
\end{split}   
\end{equation}
and, as a consequence, setting
\begin{equation}
    \mathcal{F}(T^f,T^s):=\int_\Omega \left\{[(T^f-\Ra)_+]^{2n}+\dfrac{A}{\gamma}[(T^s-\Ra)_+]^{2n}\right\}\,d\Omega,
\end{equation}
if we sum Eq. \eqref{1.1} to Eq. \eqref{1.2}, we have
\begin{equation}\label{2}
	\begin{split}
		\dfrac{1}{2n}\dfrac{d\mathcal{F}}{dt}&=-(2n-1)\int_\Omega [(T^f-\Ra)_+]^{2n-2}(\nabla T^f)^2\,d\Omega\\
  &\qquad\qquad\qquad-\dfrac{1}{\gamma}(2n-1)\int_\Omega [(T^s-\Ra)_+]^{2n-2}(\nabla T^s)^2\,d\Omega\\
		& -H\int_\Omega[(T^f-\Ra)_+-(T^s-\Ra)_+]\left([(T^f-\Ra)_+]^{2n-1}-[(T^s-\Ra)_+]^{2n-1}\right)\,d\Omega \\
		&\leq 0.
	\end{split}
\end{equation}
Inequality \eqref{2} shows that
\begin{equation}\label{key}
	\dfrac{d}{dt}\mathcal{F}^{\frac{1}{2n}}=	\dfrac{1}{2n}\dfrac{d\mathcal{F}}{dt}\mathcal{F}^{\frac{1}{2n}-1} \leq 0,
\end{equation}
hence
\begin{equation}\label{3}
	\left[\mathcal{F}(T^f,T^s)\right]^{\frac{1}{2n}}\leq \left[\mathcal{F}(T^f_0,T^s_0)\right]^{\frac{1}{2n}}.
\end{equation}
Recalling that $ \lim_{p\rightarrow \infty}\left(\int_\Omega |f|^p+|g|^p\,d\Omega\right)^{\frac{1}{p}}\leq \|f\|_{\infty}+\|g\|_\infty $, from \eqref{3} we have, assuming $ T_0^f,T^s_0\in L^\infty $, that
\begin{equation}\label{est}
\underset{\Omega}{\text{ess}\sup}\left\{(T^f-\Ra)_+\right\}\leq \underset{\Omega}{\text{ess}\sup}\left\{(T^f_0-\Ra)_+\right\}+\dfrac{A}{\gamma}\underset{\Omega}{\text{ess}\sup}\left\{(T^s_0-\Ra)_+\right\}(<\infty).
\end{equation}
Let us now consider $ T^f(\textbf{x},t)=\theta(\textbf{x},t)+T_b(z) $, $ T^s(\textbf{x},t)=\phi(\textbf{x},t)+T_b(z) $, with associated initial data $ T^f_0(\textbf{x})=\theta_0(\textbf{x})+T_b(z) $, $ T^s_0(\textbf{x})=\phi_0(\textbf{x})+T_b(z) $. Therefore, inequality \eqref{est} becomes
\begin{equation}\label{key}
\begin{split}
   		\underset{\Omega}{\text{ess}\sup}\left\{(\theta+T_b-\Ra)_+\right\}&\leq \underset{\Omega}{\text{ess}\sup}\left\{(\theta_0+T_b-\Ra)_+\right\}\\
     &\qquad\qquad\qquad+\dfrac{A}{\gamma}\underset{\Omega}{\text{ess}\sup}\left\{(\phi_0+T_b-\Ra)_+\right\}, 
\end{split}
\end{equation}
and setting
\begin{equation}\label{key}
	 \begin{split}
	\bar{\theta}_0:&= \underset{\Omega}{\text{ess}\sup}\left\{(\theta_0+T_b-\Ra)_+\right\}(<\infty),\\
	\bar{\phi}_0:&=\dfrac{A}{\gamma}\underset{\Omega}{\text{ess}\sup}\left\{(\phi_0+T_b-\Ra)_+\right\}(<\infty),
\end{split}
\end{equation}
we can conclude that
\begin{equation}\label{key}
	\| \theta+T_b-\Ra \|_{L^\infty(\Omega_1)}\leq \Gamma,
\end{equation}
i.e.
\begin{equation}\label{key}
	\theta(\textbf{x},t)+T_b(z)-\Ra\leq \Gamma,\qquad \text{a.e. on }\Omega_1,
\end{equation}
where
\begin{equation}\label{key}
	\Gamma=\begin{cases}
		\bar{\theta}_0 &\quad \text{if}\quad \phi_0\leq \Ra -T_b, \\
		\bar{\theta}_0+\bar{\phi}_0&\quad \text{otherwise},
	\end{cases} 
\end{equation}
and the Lemma is proved.

\bibliographystyle{unsrt}
\bibliography{main}

\end{document}